\documentclass[twoside,a4paper,12pt]{amsart}
\usepackage{amsmath}
\usepackage{amsfonts}
\usepackage{amssymb}
\usepackage{commath}
\usepackage{graphicx}

\providecommand{\U}[1]{\protect\rule{.1in}{.1in}}
\newtheorem{theorem}{Theorem}

\newtheorem{lemma}[theorem]{Lemma}

\newtheorem{remark}[theorem]{Remark}

\begin{document}
\title{Negative potentials and collapsing universes}
\author{Roberto Giamb\`{o}}
\email{roberto.giambo@unicam.it}
\address{School of Science and Technology, Mathematics Division, University
of Camerino, 62032 Camerino (MC) Italy.}
\author{John Miritzis}
\email{imyr@aegean.gr}
\author{Koralia Tzanni}
\email{tzanni@aegean.gr}
\address{Department of Marine Sciences, University of the Aegean, University
Hill, Mytilene 81100, Greece.}

\begin{abstract}
We study Friedmann--Robertson--Walker models with a perfect fluid matter
source and a scalar field nonminimally coupled to matter. We prove that a
general class of bounded from above potentials which fall to minus infinity
as the field goes to minus infinity, forces the Hubble function to diverge
to $-\infty$ in a finite time. This finite-time singularity theorem is true for arbitrary coupling coefficient, provided that it is a bounded function of the scalar field.
\end{abstract}

\maketitle

\section{Introduction}

The expansion history of the Universe is marked by two characteristic
phases, namely, inflation at the early stage and the present longer period
of acceleration. In both situations, scalar fields are essential ingredients
in the construction of cosmological scenarios aiming to describe the
evolution of the early and the present Universe. The accelerating phase of
the Universe requires scalar fields with non-negative potentials playing the
role of a cosmological term. However, potentials taking negative values
cannot be avoided in cosmological models (see \cite{ffkl} for motivations).
It is possible for the potential to either exhibit a negative minimum or to
be unbounded from below. In these cases, it is generally believed that the
Universe eventually collapses even if it is flat. This feature has been
observed in several specific models with negative potentials \cite{ffkl}-%
\cite{kowa}. In particular, for potentials falling to minus infinity as $%
\phi\rightarrow -\infty$, heuristic arguments were given in a recent paper 
\cite{tzmi}, indicating that a flat initially expanding Universe may
recollapse. Non-negative potentials with mathematically rigorous results,
have been studied by several authors (see for example \cite{foster,rend,gimi}%
), but there is not a corresponding rigorous treatment of negative
potentials. Collapsing models were built using homogeneous scalar field
solutions in \cite{foster, giam1, joshi, zia1} and the case where a scalar
field is coupled to a perfect fluid was studied in \cite{giam2,zia2}.

The physical reason to understand why the Universe eventually collapses when 
$V<0,$ is that in the Friedmann equation, 
\begin{equation*}
3H^{2}=\rho_{\mathrm{total}},
\end{equation*}
the positive energy density of ordinary matter, as well as the positive
kinetic energy density of the scalar field, decreases in an expanding
Universe. At some moment, the total energy density $\rho_{\mathrm{total}}$,
including the negative contribution $V\left( \phi\right) <0$, vanishes. Once
it happens, the Universe, stops expanding and enters the stage of
irreversible collapse \cite{kali}.

Of particular importance are potentials having a global positive maximum and
negatively diverging as $\phi \rightarrow \pm \infty $. These include
cosmological models in $N=2,4,8$ gauged supergravity \cite{linde1,klps1}, as
well as double exponential potentials studied by several authors \cite%
{double}; double exponential potentials with nonminimal coupling were
studied in \cite{tzmi}. The physical interest of these potentials is
described in \cite{kali}, where it is shown that if initially the field $%
\phi _{0}$ is near the value corresponding to the maximum of the potential,
it takes time $t\sim 0.7H_{0}^{-1}\ln \phi _{0}^{-1}$, until the field rolls
down from $\phi _{0}$ to the region where $V\left( \phi \right) $ becomes
negative and the Universe collapses. This time is comparable to the age of
our Universe, $H_{0}^{-1},$ and therefore it is possible that the present
Universe is into an accelerated phase, yet it will collapse in about 18
billion years. For detailed cosmological implications see \cite%
{klps,kali,kali1}.

In this paper we study the recollapse problem of scalar-field cosmological
models with negative potentials from a mathematical point of view. We assume
that the scalar field is nonminimally coupled to matter; the coupling
coefficient is assumed to be an arbitrary bounded function of the scalar
field. For motivation and more general couplings see \cite{tzmi,leon}.
Inclusion of nonminimal coupling increases the mathematical difficulty of
the analysis, yet it is important to consider nonminimal coupling in
scalar-field cosmology \cite{fuma}. As is stressed in \cite{fara1}, the
introduction of nonminimal coupling is not a matter of taste; a large number
of physical theories predicts the presence of a scalar field coupled to
matter. We consider a general class of potentials that are free to fall to $%
-\infty $ as $\phi \rightarrow -\infty $, have a global positive maximum and
go to zero from above as $\phi \rightarrow +\infty $. This class of
potentials includes the double exponential potentials mentioned above. We
show rigorously that almost always initially expanding flat universes
eventually recollapse. This result does not depend on the particular
functional form of the potential.

The plan of the paper is as follows. In the next section we write the field
equations for flat Friedmann--Robertson--Walker (FRW) models as a
constrained four-dimensional dynamical system and show a number of
preliminary results. In section \ref{sec3} we rigorously prove that a
general class of bounded from above potentials with $\lim_{\phi\rightarrow-%
\infty }V\left( \phi\right) =-\infty$, forces the Hubble function $H$ to
diverge to $-\infty$ in a finite time. Section \ref{sec4} is a brief
discussion about the classes of negative potentials studied in the
literature.

\section{Coupled scalar field models}

\label{sec2}

For homogeneous and isotropic flat spacetimes the field equations, reduce to
the Friedmann equation, 
\begin{equation}
3H^{2}=\rho +\frac{1}{2}\dot{\phi}^{2}+V\left( \phi \right) ;  \label{frie}
\end{equation}%
the Raychaudhuri equation, 
\begin{equation}
\dot{H}=-\frac{1}{2}\dot{\phi}^{2}-\frac{\gamma }{2}\rho ;  \label{ray}
\end{equation}%
the equation of motion of the scalar field, 
\begin{equation}
\ddot{\phi}+3H\dot{\phi}+V^{\prime }\left( \phi \right) =\frac{4-3\gamma }{2}%
Q\left( \phi \right) \rho ;  \label{ems}
\end{equation}%
and the conservation equation, 
\begin{equation}
\dot{\rho}+3\gamma \rho H=-\frac{4-3\gamma }{2}Q\left( \phi \right) \rho 
\dot{\phi}.  \label{conss}
\end{equation}%
An overdot denotes differentiation with respect to cosmic time $t$, $a\left(
t\right) $ is the scale factor, $H=\dot{a}/a$, is the Hubble function, and
units have been chosen so that $c=1=8\pi G.$ Here $V\left( \phi \right) $ is
the potential energy of the scalar field and $V^{\prime }\left( \phi \right)
=dV/d\phi $. Ordinary matter is described by a perfect fluid with equation
of state $p=(\gamma -1)\rho $, where $0\leq \gamma \leq 2$. The  coupling
coefficient, $Q\left( \phi \right) $, is assumed to be a non-negative,
bounded from above function, but otherwise arbitrary. Interaction terms
between the two matter components of the form $-\alpha \rho \dot{\phi}$, as
in Eq. (\ref{conss}), with a simple exponential potential and $\alpha =%
\mathrm{constant}$, were firstly considered in \cite{bico} (see also {\cite%
{bclm}}, and \cite{bapi} where $\alpha $ is an exponentially decreasing
function of $\phi $).  Interaction of the form $-\alpha \rho \dot{\phi}$,
appears naturally in scalar-tensor theories of gravity, where $Q$ is related
to the dilaton $\chi $ (see for example equation (1) in Refs. \cite{tzmi} or 
\cite{leon}), by $Q=\mathrm d\ln \chi /\mathrm d\phi $. The presence of the trace of the
energy-momentum tensor in the right-hand side of equations (\ref{ems})--(\ref%
{conss}), implies that there is no energy exchange between radiation and the
scalar field. Interaction between radiation and the scalar field is relevant
in the warm inflationary scenario \cite{bmr}, but since we are interested in
the late time evolution of the Universe, the case $\gamma =4/3$, is not
crucial in our analysis.

Although there is an energy exchange between the fluid and the scalar field,
it is easy to see that the set, $\rho >0,$ is invariant under the flow of
Eqs. \eqref{ray}--\eqref{conss}; therefore $\rho $ is nonzero if initially $%
\rho \left( t_{0}\right) $ is nonzero. This trivial physical requirement is
not satisfied if one assumes arbitrary interaction terms, cf. \cite{miri3}.

In the rest of the paper we suppose that $V(\phi)$ is a $C^{1}$ potential
such that:

\begin{enumerate}
\item $\lim_{\phi\rightarrow-\infty}V(\phi)=-\infty$ and $\lim_{\phi
\rightarrow+\infty}V(\phi)=0$.

\item The potential has a unique critical point $\phi_{m}>0,$ with $V(\phi
_{m})>0$, i.e., $\phi_{m}$ is a global maximum. Moreover $\phi_{m}$ is non
degenerate, i.e., $V^{\prime\prime}(\phi_{m})<0$.

\item There exist $\Lambda>0$ and $M<0$ such that 
\begin{equation}
V^{\prime}(\phi)\leq-\Lambda\,V(\phi),\qquad \text{for all }\phi<M. 
\label{eq:Vstima}
\end{equation}
\end{enumerate}

\begin{figure}[th]
\begin{center}
\includegraphics{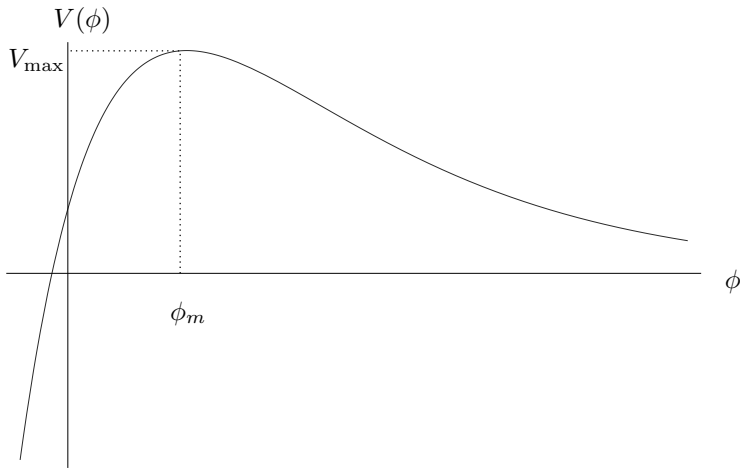}
\end{center}
\caption{A typical potential satisfying the assumptions 1,2,3. }
\label{fig1}
\end{figure}

Potentials of this type include for example double exponential potentials
studied in \cite{tzmi}. In particular, condition \eqref{eq:Vstima}
establish a bound for the growth of $|V(\phi)|$ to infinity, that must be at
most exponential.

We will also assume that the function $Q(\phi)$, that is the logarithmic derivative of the dilaton $\chi(\phi)$, is bounded for all $\phi\in \mathbb{R}$: in particular, we will suppose the existence of a constant $A$ such that
\begin{equation}\label{eq:Q}
\left\vert \frac{4-3\gamma}{2}Q(\phi)\right\vert\le A.
\end{equation}

The dynamical system \eqref{ray}--\eqref{conss} has only two finite
equilibrium points, 
\begin{equation*}
\left( \phi=\phi_{m},\dot{\phi}=0,\rho=0,H=\pm\sqrt{V_{\max}/3}\right).
\end{equation*}
They represent de Sitter and anti-de Sitter solutions and it is easy to see
that are unstable. It is known that for potentials with a maximum, the field
near the top of the potential corresponds to the tachyonic (unstable) mode
with negative mass squared \cite{linde1,klps1}. The other asymptotic states
of the system correspond to the points at infinity, $\phi\rightarrow\pm\infty
$.

It can be seen that if initially $\phi \left( 0\right) \equiv \phi _{0}<\phi
_{m}$, then there is a critical value $\dot{\phi}_{\mathrm{crit}}>0$, which
allows for $\phi $ to pass on the right of $\phi _{m}$. More precisely, in
the case of zero coupling, $Q=0,$ it is easy to show that there exists a
critical value of $\dot{\phi}$, say $\dot\phi_{\mathrm{crit}}$, such that if $\phi _{0}<\phi _{m}$ and $\dot{%
\phi}\left( 0\right) <\dot{\phi}_{\mathrm{crit}}$, then $\phi \left(
t\right) \ $ remains less than $\phi _{m}$\ for all $t\geq 0$. The argument
is similar to the mechanical analog of the motion of a particle in the
potential $V\left( \phi \right) $, according to Eq. \eqref{ems}. In the case
of nonminimal coupling, the energy density of the scalar field is not
necessarily decreasing, because there is an energy exchange between the
scalar field and the fluid. An estimation of the maximum allowable value of $%
\dot{\phi}_{0}$ can be obtained from (\ref{frie}),  supposing that initially $%
3H_{0}^{2}\leq V\left( \phi _{m}\right) =V_{\max }$. Indeed, since by Eq. %
\eqref{ray} $H$ is decreasing, 
\begin{equation}
\rho \left( t\right) +\frac{1}{2}\dot{\phi}\left( t\right) ^{2}+V\left( \phi
\left( t\right) \right) \leq V_{\max },\ \ \ \text{for all }t\geq 0,
\label{in}
\end{equation}%
which implies that $V\left( \phi \left( t\right) \right) \leq V_{\max }\ $%
for all $t\geq 0$, and therefore $\phi \left( t\right) <\phi _{m}$\ for all $%
t\geq 0$. Moreover, inequality (\ref{in}) and initial condition on $H(t)$ establish a maximum allowable value of $%
\dot{\phi}_{\mathrm{crit}},$ 
\begin{equation*}
\dot{\phi}_{\mathrm{crit}}\leq \sqrt{2V_{\max }}.
\end{equation*}%
Therefore \emph{\ if }$0<H_{0}\leq \sqrt{V_{\max }/3}$%
\emph{, then }$\phi (t)$ \emph{never crosses the maximum of} $V(\phi)$ \emph{\ throughout the evolution}.

Setting $\dot{\phi}=y$, we can write Eqs. \eqref{ray}--\eqref{conss} as an
autonomous dynamical system, 
\begin{align}
\dot{\phi} & =y,  \label{eq:1} \\
\dot{y} & =-3Hy-V^{\prime}(\phi)+\alpha(\phi)\rho,\qquad\alpha(\phi):=\frac{4-3\gamma }{2%
}Q(\phi),  \label{eq:2} \\
\dot{\rho} & =-\rho(3\gamma H+\alpha(\phi) y),  \label{eq:3} \\
\dot{H} & =-\frac{1}{2}y^{2}-\frac{\gamma}{2}\rho,   \label{eq:4}
\end{align}
subject to the constraint 
\begin{equation}
3H^{2}=\frac{1}{2}y^{2}+V(\phi)+\rho.   \label{eq:5}
\end{equation}
We recall the remarkable property of the Einstein equations that, if Eq. %
\eqref{eq:5} is satisfied at some initial time, then it is satisfied
throughout the evolution.

Our aim is to prove the following Theorem, conjectured in \cite{tzmi}.

\begin{theorem}\label{thm:main}
Let $\gamma(t)=(\phi(t),y(t),\rho(t),H(t))$ a solution to the
system \eqref{eq:1}--\eqref{eq:5} such that $\phi(t_{0})<\phi_{m}$,
$\rho(t_{0})>0$, $H(t_{0})>0$, and $y(t_{0})<\dot{\phi}_{crit}$, where
$\dot{\phi}_{crit}$ is the critical value that allows for $\phi$ to pass to
the right of $\phi_{m}$. Then $H(t)$ generically (i.e., up to a zero--measured
set of initial data) negatively diverges in a finite time:
\begin{equation}\label{eq:lem}
\exists t_{*}>0 \text{ such that } \lim_{t\to t_*^-} H(t)=-\infty.
\end{equation}
\end{theorem}

\section{Proof of Theorem \protect\ref{thm:main}}

\label{sec3} 
To prove the above theorem some preliminary results are in order. First of
all, let us prove that \emph{bounded solutions of the system \eqref{eq:1}--%
\eqref{eq:5} are non generic}.

\begin{lemma}\label{lem:1}Let
$\gamma(t)=\left(\phi(t),y(t),\rho(t),H(t)\right)$ be a bounded solution such that
$\rho(t_{0})>0$. Then $\gamma(t)\in W^{s}(q_{\pm})$, where
$W^{s}(q)$ is the stable manifold of an equilibrium point $q$
and $q_{\pm}=\left(  \phi_{m},0,0,\pm\sqrt{\tfrac{V(\phi_{m})}{3}}\right)
$ are the equilibria of the system.
\end{lemma}

\begin{proof}
Let $\mathbb{I}=[t_{0},t_{m})$, where $t_{m}$ is the supremum of the maximal
right extension of $\gamma$, and $\Omega(t)=\{\gamma(t)\,:\,t\in \mathbb{I}%
\}\cup L^{+}(\gamma)$, where $L^{+}(\gamma)$ is the positive limit set of $%
\gamma$. Equation (\ref{eq:4}) implies that $\dot{H}\leq0$ on $\Omega $. Let 
$E=\{x\in\Omega\,:\,\dot{H}=0\}=\{x\in\Omega\,:\,y=\rho=0\}$, and let $%
\eta(t)$ be a solution to the system such that $\eta(t_{0})\in E$ and $%
\eta(t)\in E$, for all $t\geq t_{0}$. It follows that $y(t)=\rho (t)=0$, for
all $t\geq t_{0}$ and, from Eq. \eqref{eq:1}, $\phi(t)=\phi_{0}$ constant
for all $t\geq t_{0}$. From Eq. \eqref{eq:2} we have that $V^{\prime
}(\phi_{0})=0$ and then $\phi_{0}=\phi_{m}$. Since $\dot{H}=0$, $H(t)$ is
constant and from Eq. \eqref{eq:5} it must be $H=\pm\sqrt{\tfrac{V(\phi_{m})%
}{3}}$. Therefore $E=\{q_{\pm}\}$, i.e. is made by the two equilibria of the
system. LaSalle invariance principle \cite{wig} and monotonicity of $H(t)$
ensure that $\gamma(t)$ converges to either $q_{+}$ or $q_{-}$, and then it
belongs to the stable manifold of one of the two equilibria.
\end{proof}

\begin{remark}\label{rem:nongen}
As a consequence of the above fact, we can show that future bounded trajectories of the system
\eqref{eq:1}--\eqref{eq:5} with $\rho(t_{0})>0$ are non generic.

Indeed, let us first observe that, using Eq. \eqref{eq:5}, we can rewrite Eq.
\eqref{eq:4} as follows
\begin{equation}
\dot{H}=-3H^{2}+\left(  1-\frac{\gamma}{2}\right)  \rho+V(\phi). \label{eq:6}%
\end{equation}
Then, let us consider the equivalent system \eqref{eq:1}--\eqref{eq:3} with Eq.
\eqref{eq:6}, and study the Jacobi matrix computed at the equilibria 
$q_{\pm}$. Since $\phi_{m}$ is a non degenerate critical point for $V(\phi)$,
we obtain that the stable manifold of $q_{+}$ is $3$--dimensional and the
stable manifold of $q_{-}$ is $1$--dimensional. In the latter case the result
straightly follows from the previous proposition. Also for the equilibrium
$q_{+}$, the result follows, taking some more care due to the fact that,
actually, Eq. \eqref{eq:5} selects a $3$--dimensional submanifold of initial
data, which anyway can be easily checked to be transversal to $W^{+}(q_{+})$
at $q_{+}$.
\end{remark}

By the above result one can expect in principle that solutions of the system %
\eqref{eq:1}--\eqref{eq:5} are generically unbounded, and our aim is now to
study their qualitative behaviour. The following fact is crucial.

\begin{lemma}\label{lem:2}
Let $\gamma(t)$ be a solution to the system
\eqref{eq:1}--\eqref{eq:5}. If there exists $ t_{1}\ge t_{0}$ and $\bar
V\in\mathbb{R}$ such that, for all $t\ge t_{1}$, $V(\phi(t))\le\bar V$, and either
$(\imath)$ $\bar V<0$, or $(\imath\imath)$
$H(t_{1})<-\sqrt{\frac{\bar V}{3}}$,
then $H(t)$ negatively diverges in a finite time, i.e. \eqref{eq:lem} holds.
\end{lemma}

\begin{proof}
To show the above, we use Eq. \eqref{eq:4} and recalling that $0\le\gamma\le2
$, we have for $t\ge t_{1}$ 
\begin{equation}  \label{eq:7}
\dot H\le\frac{\gamma}{2}\left( -3H^{2}+\bar V\right).
\end{equation}
Therefore, considering the Cauchy problem 
\begin{equation*}
\dot Z(t)=\frac{\gamma}{2}\left( -3Z(t)^{2}+\bar V\right) ,\qquad
Z(t_{1})=H(t_{1}), 
\end{equation*}
its solution $Z(t)$ is easily seen to diverge to $-\infty$ in a finite time.
The result follows from comparison theorems in ODE theory.
\end{proof}

At this point everything is set up for the proof of the main result.

\begin{proof}[Proof of Theorem \protect\ref{thm:main}]
According to Remark \ref{rem:nongen} bounded trajectories of the system %
\eqref{eq:1}--\eqref{eq:5} are non generic, we can only consider unbounded
solutions without losing genericity. Then at least one of the components of $%
\gamma(t)$ is unbounded. If $H(t)$ is unbounded, then since by Eq. %
\eqref{eq:4} $H(t)$ is decreasing, then it must be negatively unbounded, and
then Lemma \ref{lem:2} immediately gives the result, recalling that $V(\phi)$
is bounded from above. For the rest of the proof we will argue by
contradiction, and show that $H(t)$ must be necessarily unbounded.

So, suppose by contradiction that $H(t)$ is bounded. Then, assuming for sake
of simplicity that $t_{0}=0$, Eq. \eqref{eq:4} implies that there exists a
constant $K>0$ such that $|H(t)|<K$, and also 
\begin{equation}
\int_{0}^{t}\frac{1}{2}y^{2}(s)\,ds<K,\,\int_{0}^{t}\rho(s)\,ds<K,\qquad 
\text{for all } t\geq0.   \label{eq:8}
\end{equation}
Moreover, since 
\begin{equation*}
3H^{2}-V(\phi)=\frac{1}{2}y^{2}+\rho, 
\end{equation*}
and the solution must be unbounded, then either $y^{2}$ or $\rho$ (or both)
are unbounded (otherwise, $V(\phi)$ would be bounded that implies that $\phi$
is positively unbounded, which is excluded since $\phi(t)<\phi_{m}$) and
then from Eq. \eqref{eq:5} also $V(\phi)$ is negatively unbounded.

Suppose that $y^{2}$ is bounded and $\rho$ is unbounded. If $\rho$ diverges
to $\infty$ then by Eq. \eqref{eq:5} also $V(\phi)$ diverges (to $-\infty$).
Therefore, hypotheses from Lemma \ref{lem:2} are satisfied which would imply
that $H(t)$ is unbounded, contradiction. Then $\rho$ cannot diverge to $%
\infty$, and as a consequence there exists an increasing sequence $\{t_{n}\}$
such that $\rho(t_{2n})\rightarrow+\infty$ and $\rho(t_{2n-1})<\bar{\rho}$
for some fixed $\bar{\rho}$. Moreover, 
\begin{equation*}
\rho(t_{2n})-\rho(t_{2n-1})=\int_{t_{2n-1}}^{t_{2n}}\dot{\rho}(t)\,dt=\int
_{t_{2n-1}}^{t_{2n}}-\rho(t)(3\gamma H(t)+\alpha(\phi) y(t))\,dt 
\end{equation*}
and boundedness of  $y$, $H$ and -- due to \eqref{eq:Q} -- $\alpha(\phi)$, implies the existence of some positive
constant $C$ such that 
\begin{equation*}
\rho(t_{2n})-\rho(t_{2n-1})\leq C\int_{t_{2n-1}}^{t_{2n}}\rho(t)\,dt\leq
C\,K 
\end{equation*}
that is a contradiction because the left hand side diverges. Then $y^{2}$
must necessarily be unbounded, and let us now show that even in this case we
get a contradiction. To begin, observe that Eq. \eqref{eq:5} implies 
\begin{equation}
\tfrac{1}{2}y^{2}=3H^{2}-\rho-V(\phi)\leq-V(\phi)+3K^{2},   \label{eq:9}
\end{equation}
where we have also used $|H(t)|<K$, for all $t\geq0$. Let $t_{n}$ an
increasing sequence such that $y^{2}(t_{n})\rightarrow+\infty$. Then $%
\phi(t_{n})<M$ eventually, where $M$ has been defined in Eq. %
\eqref{eq:Vstima} (otherwise $V(\phi(t_{n}))$ would be bounded and then,
from Eq. \eqref{eq:9}, $y^{2}(t_{n})$ would be). Now, if $\phi(t)<M$ is
eventually satisfied for all $t$ sufficiently large (not only on the $t_{n}$%
's, namely) then $V(\phi (t))<V(M)<0$ eventually, and therefore the
hypotheses in Lemma \ref{lem:2} would be satisfied, that would mean that $%
H(t)$ is unbounded. If, on the other side, there exists an increasing
sequence $s_{n}$ such that $s_{n}<t_{n}<s_{n+1}$, $\phi(s_{n})=M$ and $%
\phi(t)<M$ in $(s_{n},t_{n})$, then it must be, by Eq. \eqref{eq:9}, 
\begin{equation*}
\frac{1}{2}y^{2}(s_{n})\leq-V(M)+3K^{2}
\end{equation*}
and therefore, using also the growth assumption made on $V$ and Eqs. %
\eqref{eq:Q}, \eqref{eq:5} and \eqref{eq:8}, 
\begin{multline*}
|y(t_{n})|\leq|y(s_{n})|+\left\vert \int_{s_{n}}^{t_{n}}\dot{y}%
(t)\,dt\right\vert \\
\leq\sqrt{2(-V(M)+3K^{2})}+3\int_{s_{n}}^{t_{n}}|H(t)y(t)|\,dt+%
\int_{s_{n}}^{t_{n}}V^{\prime}(\phi(t))\,dt+A\int_{s_{n}}^{t_{n}}%
\rho(t)\,dt \\
\leq\sqrt{2(-V(M)+3K^{2})}+\frac{3}{2}\int_{s_{n}}^{t_{n}}H^{2}(t)\,dt+%
\frac {3}{2}\int_{s_{n}}^{t_{n}}y^{2}(t)\,dt+\Lambda\int_{s_{n}}^{t_{n}}-V(%
\phi(t))\,dt+A\,K \\
\leq\sqrt{2(-V(M)+3K^{2})}+\frac{1}{2}\int_{s_{n}}^{t_{n}}3H^{2}(t)\,dt+%
\Lambda\int_{s_{n}}^{t_{n}}-V(\phi(t))\,dt+\left( 3+A\right) K \\
\leq\sqrt{2(-V(M)+3K^{2})}+(3+A)K+(1+\Lambda)%
\int_{s_{n}}^{t_{n}}(3H^{2}(t)-V(\phi(t))\,dt \\
=\sqrt{2(-V(M)+3K^{2})}+(3+A)K+\left( 1+\Lambda\right)
\int_{s_{n}}^{t_{n}}\left( \frac{1}{2}y^{2}(t)+\rho(t)\right) \,dt \\
\leq\sqrt{2(-V(M)+3K^{2})}+\left( A+5+2\Lambda\right)K ,
\end{multline*}
that is a contradiction since $|y(t_{n})|$ positively diverges. This means
that $H(t)$ cannot be bounded and therefore the result follows, as said in
the very first part of this argument, from Lemma \ref{lem:2}.
\end{proof}

\section{Discussion}

\label{sec4}

The finite-time singularity theorem of the previous section completes the analysis, carried on in \cite{tzmi}, of
the class of potentials falling to minus infinity as $\phi\rightarrow-\infty$%
, having a global positive maximum and going to zero from above as $%
\phi\rightarrow+\infty$. Assuming that the growth of $|V(\phi)|$ to infinity
is at most exponential, (see Eq. \eqref{eq:Vstima}), we proved that the
corresponding initially expanding Universes, eventually collapse in a finite
time. Our results are valid for scalar fields coupled to matter, as well as
for uncoupled models studied so far in the literature.

Our analysis does not exhaust all forms of potentials taking negative
values. The following list includes the main classes of negative potentials
encountered in the literature:

\begin{enumerate}
\item Potentials having a negative minimum. Two important examples include
the ekpyrotic potentials and those used in models of cyclic universes (see
for example \cite{linde2}, \cite{rasa}, \cite{lehne2} for reviews).

\item Bounded from below potentials with no minimum. As an example, we
mention the potentials 
\begin{equation*}
V\left( \phi \right) =V_{0}e^{-\lambda \phi }-C,\ \ \ V_{0},C,\lambda >0,
\end{equation*}%
which were considered in the context of supersymmetry theories, see for
example \cite{kali}.

\item Potentials increasing from $-\infty $ to $+\infty $, for example 
\begin{equation*}
V\left( \phi \right) =W_{0}-V_{0}\sinh \left( \lambda \phi \right) ,\ \ \ \
\ \lambda ,V_{0}>0,
\end{equation*}%
cf. \cite{gma} where an exact solution was obtained in the absence of
matter.
\end{enumerate}

From the mathematical point of view, the above examples cannot be fully
studied using the techniques exploited in this paper, and we will return to
these points in a future investigation.

\end{document}